\documentclass[letterpaper]{IEEEtran}
\usepackage{multicol}
\usepackage{lipsum}
\usepackage{array}
\usepackage[english]{babel}
\usepackage{amsmath,amssymb}
\usepackage{times}
\usepackage{relsize}
\usepackage{graphicx}
\usepackage{url}
\usepackage{booktabs}
\usepackage{multirow}
\usepackage{mparhack}
\usepackage{subfigure}
\usepackage{authblk}
\usepackage{amsthm}

\usepackage[noend]{algorithmic}

\usepackage[linesnumbered,ruled,vlined]{algorithm2e}

\usepackage{array}
\usepackage{cite}
\usepackage{eqparbox}
\usepackage{mdwmath}
\usepackage{balance}
\usepackage{epsfig}
\usepackage{xcolor}

\usepackage{mathtools}
\usepackage{bm}
\usepackage{amsfonts}
\usepackage[all]{xy}
\usepackage{etoolbox}
\usepackage{graphicx}

\DeclareMathAlphabet{\mathantt}{OT1}{antt}{li}{it}
\DeclareMathAlphabet{\mathpzc}{OT1}{pzc}{m}{it}

\DeclarePairedDelimiter\norm{\lVert}{\rVert}%

\usepackage{accents}

\newtheorem{theorem}{Theorem}

\newtheorem{lemma}[theorem]{Lemma}

\DeclareFontFamily{OT1}{pzc}{}
\DeclareFontShape{OT1}{pzc}{m}{it}%
  {<-> s * [1.1] pzcmi7t}{}
\DeclareMathAlphabet{\mathpzc}{OT1}{pzc}%
                     {m}{it}

\makeatletter
\patchcmd{\@maketitle}
  {\addvspace{0.8\baselineskip}\egroup}
  {\addvspace{-1.45\baselineskip}\egroup}
  {}
  {}
\makeatother

\title{BS-assisted Task Offloading for D2D Networks with Presence of User Mobility}
\author{Ghafour Ahani~and~\and Di Yuan}
\affil{Department of Information Technology \\ Uppsala University, Sweden\\
Emails:\{ghafour.ahani, di.yuan\}@it.uu.se
}

\begin{document}

\pagenumbering{gobble}
\maketitle
\begin{abstract}
Task offloading is a key component in mobile edge computing. Offloading a task to a remote server takes communication and networking resources. An alternative is device-to-device (D2D) offloading, where a task of a device is offloaded to some device having computational resource available. The latter requires that the devices are within the range of each other, first for task collection, and later for result gathering. Hence, in mobility scenarios, the performance of D2D offloading will suffer if the contact rates between the devices are low. We enhance the setup to base station (BS) assisted D2D offloading, namely, a BS can act as a relay for task distribution or result collection. However, this would imply additional consumption of wireless resource. The associated cost and the improvement in completion time of task offloading compose a fundamental trade-off. For the resulting optimization problem, we mathematically prove the complexity, and propose an algorithm using Lagrangian duality. The simulation results demonstrate not only that the algorithm has close-to-optimal performance, but also provide structural insights of the optimal trade-off.
\end{abstract}

\IEEEpeerreviewmaketitle
\vspace{-4mm}
\section{Introduction}
Task offloading is a key component in mobile edge
computing. Typically, tasks are offloaded to remote servers
\cite{Barbera2013,Dinh2013,shi2016} or to computing resources near to
users, e.g., base stations (BSs) \cite{Mao2017,Wang2017}. However,
these incur significant overhead in communications and networking
\cite{Chen2018}. An attractive alternative is to offload tasks to
nearby users \cite{Mtibaa2013, Mtibaa2015}. For example, a user
that currently runs on low energy
offload send its task to an idle user with energy
available for computation.

In mobility scenarios, the data of a task can be delivered
via Device-to-Device (D2D) communications as the users
move and meet each other \cite{ahani2018,Deng2018}.  However, this is
not a system-wide optimal strategy especially when some users have low
contact rates with others. In such a situation, the system can be
enhanced by letting the BSs act as relays for task distribution and
result collection. In fact, this approach enables to utilize better
the energy capacity of users. On the other hand, all tasks should not
be relayed via the BSs as this requires a large number of
communications with the BSs. Thus, \textit{which device to offload}
and \textit{how long one should wait before the BS is called for} are
both key aspects in optimal task offloading.

Looking into the literature, there are relative few works
\cite{Mtibaa2013,Mtibaa2015,Wang2014,chen2016} that considered task
offloading in mobility scenarios. The works in
\cite{Mtibaa2013,Mtibaa2015} assumed that the connection between two
users is stable during the entire offloading process. The authors of
\cite{Wang2014} considered offloading one task to nearby users with
maximization of success ratio of obtaining the result. However, none
of these studies utilized BSs as relays. The investigation in
\cite{chen2016} considered a hybrid method
where a task can be offloaded to a remote server, a BS, or a nearby
mobile user. BSs are utilized as relays for
delivering the results, however the trade-off between completion time and
the cost of using BS is not accounted for.

In this paper, we study task offloading where users can offload their
tasks to either remote servers or peer devices, possibly using the
BSs, in a mobility scenario. For each task, we define a cost related
to the completion time and processing. For offloading, a user can wait
longer time to increases the opportunity of contact and then
collecting the result via D2D, but the completion time could be quite
long. The completion time can be made shorter if a BS assists with the
offloading, but this will involve additional communications
costs. Therefore, we optimize the time before the BS is involved in
task offloading. In addition, each task has a completion time deadline
before which the result of the task must be obtained. Our aim is to
minimize the total cost of the system. Moreover, the available energy
of the users for processing is taken into consideration.
The contributions of this study are as follows.  We formulate the task
offloading problem and show how it can be effectively linearized.
We also prove mathematically the complexity of the
problem.  Next, an algorithm based on Lagrangian duality is provided
for problem solving. Our algorithm is compared to other
algorithms. Simulation results demonstrate not only that the algorithm
has close-to-optimal performance, but also provide structural insights
of the optimal trade-off.

\section{System Model and Problem Formulation}\label{sec:system}
\subsection{System Model}

In our system scenario, a set of users need to offload their tasks. We call them \textit{requesters} and the set is denoted by $\mathcal{R}=\{1,2,\dots,R\}$. The second set of users, referred to as \textit{helpers}, have energy available for task processing. The index set of helpers is denoted by $\mathcal{H}=\{1,2,\dots,H\}$. We assume that all users are within the coverage of network such that BSs can be used as relays for task distribution or result collection. Merely to simplify the presentation, we assume there is one BS. The system scenario is shown in Figure \ref{fig:system_model}.
\begin{figure}[ht!]
\centering
\includegraphics[scale=0.5]{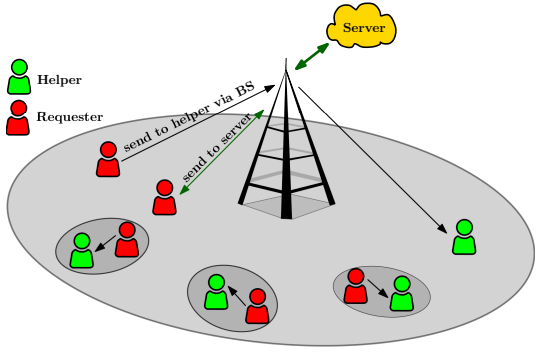}
\vspace{-3mm}
\begin{center}
  \caption{System scenario of D2D task offloading with possible BS assistance and presence of user mobility.}
  \label{fig:system_model}
  \vspace{-3mm}
\end{center}
\end{figure}

For the sake of presentation, we assume each requester has only one task. However, our formulations and algorithms can be generalized easily to a more general scenario where each requester has multiple tasks. Hereafter we use task $r$ and requester $r$ interchangeably.
The required amount of energy for processing task $r$ and one communication with the BS are denoted by $e^p_{r}$ and $e^c_r$ respectively. Each helper $h$, $h \in \mathcal{H}$, can provide at most $E_h$ amount of energy to process tasks.
Processing task $r$ by helpers and the remote server incur costs $\delta^p_{rh}$ and $\delta^p_{rN}$ respectively. These costs typically relate to the amount of energy required for computation. Each communication with the BS and remote server incurs a cost, denoted by $\delta^c_{rB}$ and $\delta^c_{rN}$ respectively. The cost of D2D communications is negligible.

The inter-contact model is widely used to characterize the mobility pattern
of mobile users~\cite{VConan2008,TDeng2017Cost}.  Hereafter, the term contact is used
to refer to the event that two users come into the communication range
of each other. The inter-contact time, that is the time between two
consecutive time point of meeting each other, for any two users
follows an exponential distribution \cite{Zhu2010}.  Hence, the number
of contacts between any two mobile users follows Poisson distribution
\cite{Sermpezis2015}. Moreover, it is assumed that the contact of user
pairs are independent.

We consider a time slotted system consisting of $T$ time slots,
$\mathcal{T}=\{1,2,\dots,T\}$, each with duration $\theta$. The
deadline of task $r$ is denoted by $d_r$.  As mentioned earlier, we
optimize the time before which a requester uses the BS for task
distribution and/or result collection. Thus for requester $r$ and
helper $h$ there is a \textit{timer} and its value is denoted by
$t_{rh}$.  The tasks are assumed to be delay tolerant, hence the
magnitude of time slot\footnote{
The magnitude of a time slot is in a range of hour.}  is considerably
larger than task processing time. Therefore, we do not account for the
processing time of the tasks. Moreover as the contact between the
helpers and requesters are stochastic, we consider the expected value
of the total cost of system.
The following five events may occur once helper $h$ is designated to task $r$:
\begin{enumerate}
\item
They meet at least twice before $t_{rh}$, then the task is collected and result is obtained,
both via D2D.
\item They meet exactly once before $t_{rh}$, then the task is collected, and they meet at least once again between $t_{rh}$ and $d_r$, then the result is obtained. This case also uses D2D communications twice.
\item They do not meet before $t_{rh}$, but they meet at least once between $t_{rh}$ and $d_r$. Then the BS is involved to deliver the task to the helper (with two communications: requester $\rightarrow$ BS, and BS $\rightarrow$ helper) and the result is obtained via D2D communications.
\item They meet exactly once before $t_{rh}$, however, they do not meet after this time point until $d_r$. Then the task is given to the helper via D2D communications and the result is obtained via the BS (with two communications: helper $\rightarrow$ BS, and BS $\rightarrow$ requester).
\item
    They do not meet at all within $d_r$. In this case, the task is sent to the server for processing\footnote{In this case, the BS also can be used for both the task distribution and result collection, but we do not account for this solution because it involves four communications with the BS.}.
\end{enumerate}

There is a cost associated with task completion time defined as the starting time point until the requester obtains the task's result. We introduce a cost function $f(\cdot)$ for which $f(t)$ is the cost for a completion time of $t$ slots. For the events above, we will derive the total expected task completion cost including the task completion time and the communication if applicable.

\vspace{-3mm}
\subsection{Cost Model}
Denote by binary variable $x_{rh}$ representing if
requester $r$ offloads its task to helper $h$.  The corresponding
variable for task offloading to the server is denoted by $x_{rN}$.
Denote by $P(M_{rh}^{[k,l]}=n)$ the probability that requester $r$
meets helper $h$ exactly $n$ times during time slots $k$ to $l$. Note
that when $k=l$, it is the probability of having $n$ contacts within
time slot $k$. For special case $k<l$, there are two cases, i.e.,
$n>0$ and $n=0$. Intuitively, their corresponding probabilities are
defined to zero and one. The probability $P(M_{rh}^{[k,l]}\ge n)$ is
defined in a similar way.  Here, $M_{rh}^{[k,l]}$ follows a Poisson
distribution with mean $\lambda_{rh}(l-k+1)\theta$, where
$\lambda_{rh}$ represents the average number of contacts per unit
time. Denote by $\Pi_{rh}^{(i)}, i\in \{1,2,3,4,5\}$, the probability
that event $i$ occurs and the expected cost of event $i$, $i\in
\{1,2,3,4,5\}$, is denoted by $\Delta_{rh}^{(i)}$.

The cost of assigning the task $r$ to helper $h$ originates from
waiting time before task completion, communications with BS (if
applicable), and task processing. The associated expected cost of each
event is derived and shown in Table \ref{tab:events}. For events $1$
and $2$, the first and second terms are the expected costs related to
task completion time and processing respectively. For events $3$ and
$4$, the first, second, and third terms are the expected costs related
to task completion time, processing, and communication with BS. For
event $5$, we have costs of processing and communications with
server.
\begin{table*}[t]
  \centering
    \caption{Expected costs and probabilities of events.}
  \label{tab:events}
  \vspace{-3mm}
  \begin{tabular}{lll}
  \hline \hline
  Event&  \multicolumn{1}{c}{Expected cost}&\multicolumn{1}{c}{Probability}\\
  \hline\hline
   1&
   $\!\begin{aligned}[t]\Delta_{rh}^{(1)}(t_{rh})=\sum_{k=0}^{t_{rh}}{f(k)
\Big(P(M_{rh}^{[1,k]}\ge2)-P(M_{rh}^{[1,k-1]}\ge2)\Big)}
+\Pi_{rh}^{(1)}(t_{rh})\delta^p_{rh}\end{aligned}$ & $\Pi_{rh}^{(1)}(t_{rh})=P(M_{rh}^{[1,t_{rh}]}\ge2)$  \\

   2&
   $\!\begin{aligned}[t]
    \Delta_{rh}^{(2)}(t_{rh})=&
 \sum_{k=t_{rh}+1}^{d_{r}}f(k)P(M_{rh}^{[1,t_{rh}]}=1) \Big(P(M_{rh}^{[t_{rh}+1,k]}\ge1)-\\
   &P(M_{rh}^{[t_{rh}+1,k-1]}\ge 1)\Big)+\Pi_{rh}^{(2)}(t_{rh})\delta^p_{rh}
    \end{aligned}$
    &
    $\Pi_{rh}^{(2)}(t_{rh})=P(M_{rh}^{[1,t_{rh}]}=1)P(M_{rh}^{[t_{rh}+1,d_{r}]}\ge1)$ \\

    3&
       $\!\begin{aligned}[t]
    \Delta_{rh}^{(3)}(t_{rh})=&\sum_{k=t_{rh}+1}^{d_{r}}f(k) P(M_{rh}^{[1,t_{rh}]}=0) \Big(P(M_{rh}^{[t_{rh}+1,k]}\ge1)-\\
&P(M_{rh}^{[t_{rh}+1,k-1]}\ge1)\Big)+\Pi_{rh}^{(3)}(t_{rh})\Big(\delta^p_{rh}+2\delta^c_{rB}\Big)
    \end{aligned}$
&
$\Pi_{rh}^{(3)}(t_{rh})=P(M_{rh}^{[1,t_{rh}]}=0)P(M_{rh}^{[t_{rh}+1,d_{r}]}\ge1)$\\

  4&
   $\Delta_{rh}^{(4)}(t_{rh})= \Pi_{rh}^{(4)}(t_{rh})\Big(f(t_{rh})+\delta^p_{rh}+2\delta^c_{rB}\Big)$
&
$\Pi_{rh}^{(4)}(t_{rh})=P(M_{rh}^{[1,t_{rh}]}=1)P(M_{rh}^{[t_{rh}+1,d_{r}]}=0)$\\

  5
  &
  $\Delta_{rh}^{(5)}(t_{rh})= \Pi_{rh}^{(5)}(t_{rh})\Big(\delta^p_{rN}+2\delta^c_{rN}\Big)$
&$\Pi_{rh}^{(5)}(t_{rh})=1-\sum_{i=1}^{4}\Pi_{rh}^{(i)}(t_{rh})$
\\
\hline
\hline
  \end{tabular}
\end{table*}
Thus, the total expected cost for using helper $h$ for task $r$ is:
\vspace{-2mm}
\begin{equation}
\begin{split}
\Delta_{rh}(t_{rh})=\sum_{i=1}^{5}\Delta_{rh}^{(i)}(t_{rh})
\end{split}
\end{equation}

Hence, the overall cost for all helpers and requesters is:

\vspace{-3mm}
\begin{equation}
\begin{split}
\text{Cost}({\bm{x},\bm{t}})=\sum_{r \in \mathcal{R}}\sum_{h \in \mathcal{H}}\Delta_{rh}(t_{rh})x_{rh}+\sum_{r \in \mathcal{R}}(2\delta^c_{rN}+\delta^p_{rN}) x_{rN}\label{eq:overalcost}
\end{split}
\end{equation}
where $\bm{x}$ and $\bm{t}$ are two matrices of dimensions
$R\times(H+1)$ and $R\times H$, respectively, representing
the offloading and timer variables.
Note that the cost function is
highly nonlinear, but we prove in Section \ref{sec:alg_design} that
this function can be linearized and the optimal value of timers
$t_{rh}$, $r \in\mathcal{R}, h\in \mathcal{H}$, can be preprocessed.

\vspace{-1mm}
\subsection{Energy Consumption on Helpers}
The energy consumed on a helper consists of those for processing and
communications, whereas that for
D2D communications is negligible in
comparison. Therefore, we consider the expected consumed
energy. Hence, for requester $r$ we have:
\vspace{-1mm}
\begin{equation}
\begin{split}
e_r=& \Pi_{rh}^{(1)}(t_{rh})e^p_r+\Pi_{rh}^{(2)}(t_{rh})e^p_r+\\
&\Pi_{rh}^{(3)}(t_{rh})(e^p_r+e^c_r)+\Pi_{rh}^{(4)}(t_{rh})(e^p_r+e^c_r)
\end{split}
\end{equation}
\vspace{-5mm}

\subsection{Problem Formulation}
The problem is formulated as follows:
\vspace{-5mm}
\begin{figure}[ht]
\begin{subequations}\label{form}
\begin{alignat}{2}
\quad &
{\min_{\substack{\bm{x} \in \{0,1\}^{R \times (H+1)}, \bm{t}\in \{0,1\}^{R \times H}}}~\text{Cost}(\bm{x},\bm{t})}
\label{F_e} \\
\text{s.t}. \quad
&\sum_{h\in \mathcal{H}} x_{rh}+x_{rN} = 1,r\in \mathcal{R} \label{const:1helper}\\
& \underset{r\in \mathcal{R}}\sum e_{r} x_{rh}\le E_h,h\in \mathcal{H} \label{const:totalenergy}
\end{alignat}
\end{subequations}
\end{figure}
\vspace{-4mm}

Constraints (\ref{const:1helper}) indicate that a requester must offload its task to either a helper or to the server. Constraints (\ref{const:totalenergy}) respect the available energy of helpers.
\section{Complexity Analysis}\label{sec:complexity}
\begin{theorem}\label{theorem_NP}
The task offloading problem is $\mathcal{NP}$-hard.
\end{theorem}
\begin{proof}
We adopt a polynomial-time reduction from the Knapsack problem of $N$
items having weights $\{e_1,e_2,\dots,e_N\}$, values
$\{g_1,g_2,\dots,g_N\}$, and capacity $E$.  Our reduction is as
follows. We have one helper, i.e., $\mathcal{H}=\{1\}$, with total
available energy $E$. There are $N$ requesters, i.e.,
$\mathcal{R}=\{1,2,\dots,N\}$. The expected amount of energy for
processing task $r$ is $e_r$. We set
$\lambda_{rh}=\ln(\frac{1}{\epsilon})$ for all requesters and helpers
where $\epsilon$ is a small positive number. Therefore, the requesters
and helpers meet at least once with probability $1-\epsilon$. Also, we
use $f(t)=0$ and set the same deadlines for all requesters, i.e.,
$d_{r}=d_{r^{\prime}}, r,r^{\prime} \in \mathcal{R}$. By construction
the optimal timer value is $d_r$ for all requesters. The number of
time slots is set to the deadline of tasks, i.e., $T=d_r$. The costs
for processing task $r$ by the helper and the server is set to
$\delta^p_{rh}=0$ and $\delta^p_{rN}=0$. The cost of communications
with BS and server is set to $\delta^c_{rB}=0$ and
$\delta^c_{rN}=\frac{g_r}{2}$. This setting results in the overall
completing cost for task $r$ by the helper and the server as $0$ and
$g_r$ respectively, i.e., $\Delta_{rh}=0$ and
$2\delta^c_{rN}+\delta^p_{rN}=g_r$. Consequently, if task $r$ is
offloaded to the helper, its gain is
$2\delta^c_{rN}+\delta^p_{rN}-\Delta_{rh}=g_r$. By construction, the
optimum to our problem solves the Knapsack problem instance.
As the Knapsack problem is $\mathcal{NP}$-hard, the conclusion follows.
\end{proof}

\section{Algorithm Design}\label{sec:alg_design}
The cost in equation (\ref{eq:overalcost}) has a rather complicated structure because of the nonlinearity. However, in the following we provide a structure insight stating that for each pair task $r$ and candidate helper $h$, the optimal value of $t_{rh}$ can be preprocessed. This enables us to reformulate the cost function as a linear function without loss of optimality.
\begin{lemma}
\label{Lemma1}
For any pair $r$ and $h$, the optimal value of timer $t_{rh}$ can be obtained with linear complexity.
\end{lemma}

\begin{proof}
For each possible value of $t_{rh}$ from $0$ to
$d_r$, the value of $\Delta_{rh}(t_{rh})$ can be computed
in linear time, because $\Delta_{rh}(t_{rh})$ is the sum of costs of
the five possible events. The cost of each of them involves
calculating the probabilities and cost of completion time. The
probabilities can be obtained in $O(1)$ via formula $\frac{\lambda^k
e^{-\lambda}}{k!}$ and the completion time cost can be obtained in
$O(T)$ as there is maximum $T$ time slots. Thus the overall complexity
is $O(T)$.  Furthermore, the value of $\Delta_{rh}(t_{rh})$ is
independent from the other pairs. These together enable us to obtain
the optimal value of $t_{rh}$ by taking $\min$ operator over all
possible values, i.e., $\Delta^*_{rh}=\underset{t_{rh} \in
\{0,1,\dots,d_r\}~~~}{\min~\{\Delta_{rh}(t_{rh})\}}$.
\end{proof}
By Lemma \ref{Lemma1}, the objective function is linearized below.

\vspace{-7mm}
\begin{alignat}{2}\label{linear_formulation}
\quad &
\min_{\substack{\bm{x} \in \{0,1\}^{R \times (H+1)}}} \sum_{r \in \mathcal{R}}\sum_{h \in \mathcal{H}}\Delta^*_{rh}x_{rh}+\sum_{r \in \mathcal{R}}\delta_{rN} x_{rN} \\
\text{s.t}. \quad
&\text{(\ref{const:1helper}), (\ref{const:totalenergy})} \nonumber
\end{alignat}

\vspace{-7mm}
\subsection{Lagrangian Relaxation}

We apply Lagrangian relaxation to (\ref{const:totalenergy}).
Denote by $u_h$, $h\in \mathcal{H}$, the corresponding Lagrange multipliers. We have the following Lagrangian relaxation:

\vspace{-3mm}
\begin{alignat}{2}\label{DS}
\begin{split}
L(\bm{u})&=\min_{\substack{\bm{x} \in \{0,1\}^{R \times (H+1)}}}~\sum_{r \in \mathcal{R}}\sum_{h \in \mathcal{H}}\Delta^*_{rh}x_{rh}+\\
&\sum_{r \in \mathcal{R}}(2\delta^c_{rN}+\delta^p_{rN}) x_{rN}+\sum_{h \in \mathcal{H}}u_h(\sum_{r\in \mathcal{R}}e_{r} x_{rh}-E_h)\\
&~~~\text{s.t}. \quad \text{(\ref{const:1helper})}
\end{split}
\end{alignat}

The above problem is polynomial-time solvable as the only constraint \text{(\ref{const:1helper})} states that the task has to be assigned to either a helper or the server. Therefore, the optimal is to pick the helper or the server that minimizes the expected cost.
\vspace{-2mm}
\subsection{Subgradient Optimization}
The Lagrangian dual problem is
$v^{*}=\max_{\bf{u}\ge 0} L(\bf{u})$,
where $\bm{u}=[u_1,\dots,u_H]$. A subgradient, $\bm{d}=[d_1,\dots,d_H]$, to the concave function $L(\bm{u})$ can be obtained as:
\begin{alignat}{2}
d_h=\sum_{r\in \mathcal{R}}e_{r} \bar{x}_{rh}-E_h,h \in \mathcal{H}, \label{eq:search_direction}
\end{alignat}
where $\{\bar{x}_{rh}, r \in \mathcal{R}, h \in \mathcal{H}\}$ is obtained from the optimal solution
to $L(\bm{u}$ for the given $\bm{u}$.
The dual problem can be solved with subgradient optimization, described in Algorithm \ref{alg:sub}. In the algorithm, $K$ is the maximal allowed number of iterations, $\b{$v$}$ and $\bar{v}$ denote the best known lower and upper bounds on $v^*$. Any feasible solution yields upper bounds. Initially, we use $\bar{v}=\sum_{r \in \mathcal{R}}(2\delta^c_{rN}+\delta^p_{rN})$.
We user the following formula to calculate the step \cite{POLYAK1969}:
\begin{alignat}{2}
t^{(k)}=\max\{0,\eta \frac{\bar{v}-g(\bm{u}^{(k)})}{\norm{\bm{d}^{(k)}}^2}\} \text{ with } 0<\eta<2
 \label{eq:step_length}
\end{alignat}

\begin{algorithm}
\caption{Lagrangian-based Algorithm}
\label{alg:sub}
\begin{algorithmic}[1]
\STATE Choose a starting point $\bm{u}^{(1)}$, choose $\epsilon_1>0$ and $\epsilon_2>0$, $\b{$v$}\leftarrow-\infty$, $\bar{v} \leftarrow\sum_{r \in \mathcal{R}}(2\delta^c_{rN}+\delta^p_{rN})$
\REPEAT
 \STATE Solve (\ref{DS}), yielding $L(\bm{u}^{(k)})$ and $\bar{\bm{x}}$\\
\lIf{$L(\bm{u}^{(k)})>\b{$v$}$}
{
    $\b{$v$}\leftarrow L(\bm{u}^{(k)})$
}

\STATE Make an attempt to modify $\bar{\bm{x}}$ to a feasible solution, and possibly update $\bar{v}$ \label{makefeasiblesol}
\STATE Calculate search direction $\bm{d}^{(k)}$ and step length $t^{(k)}$ using formula (\ref{eq:search_direction}) and (\ref{eq:step_length}) respectively

 \STATE Update $\bm{u}^{(k+1)}=\bm{u}^{(k)}+t^{(k)} \bm{d}^{(k)}$
 \STATE $k \leftarrow k+1$
\UNTIL{$(k > K~\text{or}~\norm{\bm{d}^{(k)}}\le \epsilon_1~\text{or}~ \norm{\bm{u}^{(k)}-\bm{u}^{(k-1)}}\le \epsilon_2)$}

\end{algorithmic}
\end{algorithm}
\vspace{-2mm}
We carry out Step \ref{makefeasiblesol} as follows.
For each helper having its energy constraint violated, we reassign some of the allocated tasks to other helpers in ascending order of cost.

\vspace{-1mm}
\section{Performance Evaluation}\label{sec:performance}

For performance comparison, we consider two intuitive task offloading
strategies based on the expected cost and contact rates
respectively. For these two strategies the tasks are allocated to
helpers in descending order of expected completion cost and contact
rate, respectively. After allocating a task, the residual energy of
helpers will be updated. This process will be repeated until each task
is assigned to either a helper or the server.

The energy available of each helper is generated randomly within
interval $[1000,3000]$ Joule (J).  The experiments in
\cite{Passarella2013} have shown that $\lambda_{rh}$ follows Gamma
distribution. Here, we use Gamma distribution $\Gamma(0.5,1)$. The
energy required to process a task depends on two factors: the size of
data and type of workload \cite{Kwak2015}, and the number of CPU
cycles for processing one bit varies by workload
type~\cite{Kwak2015,Miettinen2010}.  We generate the tasks with data
size within interval $[0.5,5]$ MB and assign them workload such that
CPU cycles per bit is in the range
$[2000,37000]$. We consider $\frac{1}{730\times
10^6}$ J and $1.42\times 10^{-7}$ J for energy consumption of one CPU
cycle and one bit data transmission respectively
\cite{Miettinen2010}. The processing cost of task $r$ on helper $h$
and the remote server are set to $\delta^p_{rh}=e^p_r$ and
$\delta^p_{rN}=10e^p_r$. The other parameters are set as follows:
$\delta^c_{rB}=100e^c_r$, $\delta^c_{rN}=1000 e^c_r$, $f(t)=\alpha
t^2$ where $\alpha$ is a weighting factor, $T=24$, and $\theta=1$
hour. The deadlines of tasks are generated randomly within range
$[1,24]$ time slots. A task with more required energy has a longer
deadline. All simulation results are obtained by averaging over $100$
instances.

Figures \ref{fig:impcatH} and \ref{fig:impcatAlpha} show the impact of the number of helpers $R$ and weighting factor $\alpha$ on the expected total cost. In Figure \ref{fig:impcatH}, as expected, the cost decreases with respect to $H$.  For $H=3$, the performance gaps of cost-based and contact-based strategies with respect to the Lagrangian-based algorithm are $19\%$ and $5\%$ respectively, and the values grow to $29\%$ and $80\%$ for $H=7$. The reason for the increase is that the available energy is limited when $H=3$, thus most of the tasks are offloaded to the server, no matter which strategy or algorithm is used. But, when the number of helpers increases to $H=7$, the Lagrangian-based algorithm manages to utilize the energy of helpers to accommodate more tasks, whereas the two other strategies are less optimal in this regard. In addition, the solution from the Lagrangian-based algorithm is about $20\%$ from the lower bound of global optimum. This manifests that our algorithm produces close-to-optimal solutions.

In Figure~\ref{fig:impcatAlpha}, we observe that with the increase of $\alpha$, the overall expected cost increases. This is expected as a higher $\alpha$ means a growth in a coefficient in the objective function, whereas the solution space remains unchanged.
The contact-based strategy performs better than the cost-based one for large values of $\alpha$. The reason is that
the expected cost of each event basically consists of two main parts: the cost related to the expected completion time and the cost related to the processing and communications. The former depends on the weighting factor $\alpha$ and the contact rates between the requesters and helpers. Thus, larger $\alpha$ gives more emphasis on the contact rates, and consequently the contact-based strategy shows better performance when $\alpha$ increases.
Furthermore, the Lagrangian-based algorithm consistently and significantly outperforms the cost-based and contact-based task allocation strategies.

In Figure~\ref{fig:cdf}, the x-axis is the relative timer values with
respect to the deadline, before the BS is used, and the curves show
the percentages of requesters setting their timers being at most the
values of the x-axis, for various values of $\alpha$. We can see that
when parameter $\alpha$ increases, there are more requesters using
shorter timers. For example for $\alpha=0.1$, about $50\%$ of
requesters will use the BS at time point zero, while this percentage
for $\alpha=0.00016$ decreases to almost zero.  These results provide
structural insights of using D2D communications versus the BS as well as
the resulting cost trade-off, in relation to the amount of emphasis
put on task completion time.

\begin{figure}[ht!]
\centering
\includegraphics[scale=0.5]{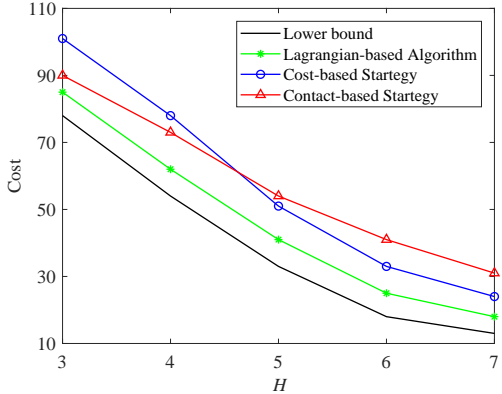}
\vspace{-3mm}
\begin{center}
\vspace{-3mm}
  \caption{Impact of $H$ on cost when $R=15$, $T=24$ , $\theta=1$ and $f(t)=0.004t^2$.}
  \label{fig:impcatH}
  \vspace{-3mm}
\end{center}
\end{figure}
\vspace{-3mm}
\begin{figure}[ht!]
\centering
\includegraphics[scale=0.5]{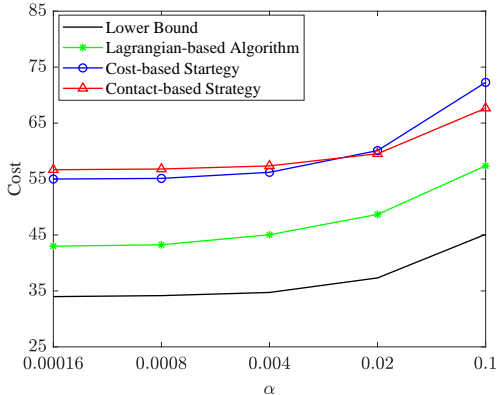}
\vspace{-3mm}
\begin{center}
\vspace{-3mm}
  \caption{Impact of $\alpha$ on cost when $R=15$, $H=5$, $T=24$, $\theta=1$, and $f(t)=\alpha t^2$.}
  \label{fig:impcatAlpha}
  \vspace{-3mm}
\end{center}
\end{figure}

\vspace{-3mm}
\begin{figure}[ht!]
\centering
\includegraphics[scale=0.5]{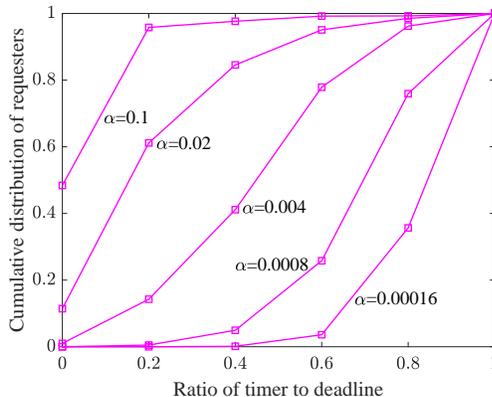}
\vspace{-3mm}
\begin{center}
\vspace{-3mm}
  \caption{Impact of $\alpha$ on the timers, i.e., the amount of time that a user waits before the BS is called for assistance, with respect to deadlines when $R=15, H=5$, $T=24$, $\theta=1$, and $f(t)=\alpha t^2$.}
  \label{fig:cdf}
  \vspace{-3mm}
\end{center}
\end{figure}
\vspace{-3mm}

\section{conclusions}
We have studied a task offloading problem with presence
of user mobility and possible assistance of BS as relay.
For this optimization problem, we have provided structural
insight, complexity analysis, and a solution algorithm.
Simulation results manifested that our algorithm
has a small gap with the optimal solutions and outperforms the other
two strategies, i.e., cost-based and contact-based strategies. The
future work plan is to investigate a more hierarchical task offloading
architecture for mobility scenarios.

\vspace{-3mm}
\bibliographystyle{IEEEtran}
\bibliography{IEEEabrv,ForIEEEBib}

 \end{document}